\documentclass[conference]{IEEEtran}
\usepackage{cite}
\usepackage{amsmath,amssymb,amsfonts}
\usepackage{algorithmic}
\usepackage{graphicx}
\usepackage{tikz}
\usepackage{subcaption}
\usetikzlibrary{shapes,arrows}
\usepackage{multicol,lipsum}
\usepackage{float}
\usepackage{lipsum}
\usepackage{mathtools}
\usepackage{cuted}
\usepackage{hyperref}
\usepackage{textcomp}
\usepackage{amsmath}
\usepackage{esvect}
\usepackage{bm}

\usepackage{amsmath,bbm,amssymb,amsfonts,amstext,amsopn}

\usepackage{lipsum,amsmath,multicol}
\usepackage{amsthm}
\usepackage{amssymb}
\usepackage{epstopdf}

\usepackage{amsthm}
\usepackage{cite}
\usepackage{balance}
\allowbreak
\usepackage{url}
\usepackage{epsfig}
\usepackage{setspace}
\usepackage{stmaryrd}
\usepackage{psfrag}	
\usepackage{multirow}
\usepackage{float}
\usepackage{amsthm}

\newcommand{\subparagraph}{}
\usepackage[compact]{titlesec}
\usepackage[english]{babel}
\usepackage[process=auto]{pstool}
\usepackage{amsthm}
\theoremstyle{remark}
\newtheorem{remark}{Remark}
\usepackage{etoolbox}
\usepackage{algorithm}
\usepackage[process=auto]{pstool}
\usepackage{epsfig}
\usepackage{caption}
\usepackage{algorithmic}
\usepackage{xcolor}

\newtheorem{lem}{Lemma}

\begin{document}
\title{Joint Uplink-Downlink Resource Allocation for OFDMA-URLLC MEC Systems} 
\author{ Walid R. Ghanem,
Vahid Jamali, Qiuyu Zhang, and Robert Schober  \\
Friedrich-Alexander-University Erlangen-Nuremberg, Germany \quad}
\maketitle
\begin{abstract}
In this paper, we study resource allocation algorithm design for multiuser orthogonal frequency division multiple access (OFDMA) ultra-reliable low latency communication (URLLC) in mobile edge computing (MEC) systems. To achieve the stringent end-to-end delay and reliability requirements of URLLC MEC systems, we propose joint uplink-downlink resource allocation and finite blocklength transmission. Furthermore, we propose a partial time overlap between the uplink and downlink frames to minimize the end-to-end delay, which introduces new time causality constraints. Then, the proposed resource allocation algorithm is formulated as an optimization problem for minimization of the total weighted transmit power of the network under constraints on the minimum quality-of-service regarding the number of computed URLLC user bits within the maximum allowable computing time, i.e., the end-to-end delay of a computation task. Due to the non-convexity of the optimization problem, finding the globally optimal solution entails a high computational complexity which is not tolerable for real-time applications. Therefore, a low-complexity algorithm based on successive convex approximation is proposed to find a high-quality sub-optimal solution. Our simulation results show that the proposed resource allocation algorithm design facilitates the application of URLLC in MEC systems, and yields significant power savings compared to a benchmark scheme.
\end{abstract}
\section{Introduction} 
Future wireless communication networks have several system design objectives including high data rates, reduced latency, and massive device connectivity. One important objective is to enable ultra-reliable low latency communication (URLLC). URLLC will be widely adopted for mission-critical applications such as remote surgery, factory automation, autonomous driving, tactile Internet, and augmented reality to enable real-time machine-to-machine and human-to-machine interaction \cite{popski}. URLLC imposes strict quality-of-service (QoS) constraints including a very low latency (e.g., 1 $\textrm{ms}$) and a low packet error probability (e.g., $10^{-6}$). 

Recently, significant attention has been devoted to studying and developing resource allocation algorithms enabling URLLC. In particular, optimal power allocation in a multiuser time division multiple access (TDMA) URLLC system was considered in \cite{optimal,convexfinite}. Moreover, resource allocation for orthogonal frequency division multiple access (OFDMA)-URLLC systems was studied in \cite{chsejoint,ghanem1,gha}. However, the existing resource allocation schemes in \cite{optimal,convexfinite,ghanem1,gha} focused only on communication while computation was not considered. Nevertheless, devices in mission-critical applications will also generate tasks that require computation within a given time. Therefore, resource allocation algorithm design for efficient computation in URLLC systems has to be investigated.

A promising solution to enable efficient and fast computation for URLLC devices is mobile edge computing (MEC). MEC enhances the battery lifetime and reduces the power consumption of users with delay-sensitive tasks. By offloading these tasks to nearby MEC servers, the power consumption and computation time at the local users can be considerably reduced at the expense of the power required for the data transmission for offloading. Thus, efficient resource allocation algorithm design is paramount for MEC for optimization of the available resources (e.g., power and bandwidth) while guaranteeing the maximum delay for the computation tasks. Existing resource allocation algorithms for MEC designs, such as \cite{energymec,Zhoumec}, were designed based on Shannon's capacity formula. In particular,  the authors in \cite{energymec} studied energy-efficient resource allocation for MEC, while computation rate maximization was considered in \cite{Zhoumec}. However, if the resource allocation design for URLLC MEC systems is based on Shannon's capacity formula, the reliability of the offloading and downloading processes cannot be guaranteed. To cope with this issue, recent works applied finite blocklength transmission (FBT)\cite{Polyanskiy} for resource allocation algorithm design for URLLC MEC systems. In particular, the authors in \cite{9048917} studied binary offloading in single-carrier TDMA systems. However, single-carrier systems suffer from poor spectrum utilization and require complex equalization at the receiver. In \cite{DeepMEC}, the authors investigated the minimization of the normalized energy consumption for OFDMA. However, the algorithm proposed in \cite{DeepMEC} assumes that the channel gain is identical for different sub-carriers which may not be realistic for broadband wireless channels. Moreover, the resource allocation algorithms proposed in \cite{DeepMEC} are based on a simplified version of the general expression for the achievable rate for FBT \cite{Polyanskiy}. Furthermore, the existing MEC designs, such as \cite{energymec,uavmec}, do not take into account the size of the computation result of the tasks and do not consider the communication resources consumed for downloading of the processed data by the users. Nevertheless, the size of the processed data can be large for applications such as augmented reality URLLC. To the best of the authors' knowledge, joint uplink-downlink resource allocation for OFDMA-URLLC MEC systems has not been considered in the literature, yet.

Motivated by the above discussion, in this paper, we propose a novel power-efficient joint uplink-downlink resource allocation algorithm design for multiuser OFDMA-URLLC MEC systems. To reduce the end-to-end delay of the uplink and downlink transmission while efficiently exploiting the available spectrum, we propose a partial time overlap between the uplink and downlink frames which introduces new causality constraints. Then, the resource allocation algorithm design is formulated as an optimization problem for the minimization of the total weighted power consumed by the base station (BS) and the users subject to QoS constraints for the URLLC users. The QoS constraints include the minimum required number of bits computed within the maximum allowable time for computation, i.e., the maximum end-to-end delay of each user. The formulated optimization problem is a non-convex mixed-integer problem that is difficult to solve globally. Thus, we develop a low-complexity sub-optimal algorithm based on successive convex approximation (SCA) in order to find a locally optimal solution.

\textit{Notation}: Lower-case letters $x$ refer to scalar numbers, while bold lower-case letters $\mathbf{x}$ represent vectors. $(\cdot)^{T}$ denotes the transpose operator. 
$\mathbb{R}^{N \times 1}$ represents the set of all $N \times 1$ vectors with real valued entries. The circularly symmetric complex Gaussian distribution with mean $\mu$ and variance $\sigma^{2}$ is denoted by $\mathcal{CN}(\mu,\sigma^{2})$, $\sim$ stands for ``distributed as", and $\mathcal{E}\{\cdot\}$ denotes statistical expectation. $\nabla_{\mathbf{x}}f(\mathbf{x})$ denotes the gradient vector of function $f(\mathbf{x})$ and its elements are the partial derivatives of $f(\mathbf{x})$.
\section{System and Channel Models}
In this section, we present the considered system and channel models for OFDMA-URLLC MEC systems.
\subsection{System Model}
We consider a single-cell multiuser MEC system which comprises a BS and $K$ URLLC users indexed by $k =\{1,\dots,K\}$, cf. Fig.~\ref{model}. All transceivers have single antennas. The system employs frequency division duplex (FDD)\footnote{In FDD systems, different frequency bands are assigned to uplink and downlink.}. Thereby, the total bandwidth $W$ is divided into two bands for uplink and downlink having bandwidths $W^{u}$ and $W^{d}$, respectively. The bandwidths for uplink and downlink are further divided into $M^{u}$ and $M^{d}$ orthogonal sub-carriers indexed by $m^{u} =\{1,\dots,M^{u}\}$ and $m^{d} =\{1,\dots,M^{d}\}$, respectively. The bandwidth of each sub-carrier is $BW_{
	s}$. Thus, the symbol duration is $T_{s}=\frac{1}{BW_{s}}$. The uplink and downlink frames are divided into $N^{u}$ time slots indexed by $n^{u} =\{1,\dots,N^{u}\}$ and $N^{d}$ time slots indexed by $n^{d} =\{1,\dots,N^{d}\}$, respectively. Moreover, each time slot contains one OFDM symbol. The downlink transmission starts after $\tau$ time slots. Thus, uplink and downlink transmission overlap in $\bar{O}=N^{u}-\tau$ time slots. The value of $\tau$ is a design parameter. On the one hand, if $\tau$ is chosen too small, the users' information bits to be computed may have not yet arrived at the BS and hence the downlink resource is wasted. On the other hand, if $\tau$ is chosen too large, the computed bits at the BS have to wait before being transmitted to the users, which increases the end-to-end delay, see Fig. \ref{model}. Each user has one computation task ($B_{k}$, $D_{k}$) that needs to be processed, where $B_{k}$ is the task length in bits and $D_{k}$ is the required time for computation in time slots. Moreover, we assume that all users offload their tasks to the MEC server. The maximum transmit power of the BS is $P_{\text{max}}$, while the maximum transmit power of each user in the uplink is $P_{k,\text{max}}$. 

In order to facilitate the presentation, in the following, we use superscript
$j \in \{u,d\}$ to denote uplink $u$ and downlink $d$.
\begin{remark}
	The power and time consumed for channel estimation and resource allocation are constant and will not affect the validity of the proposed resource allocation algorithm. For simplicity of illustration, they are neglected in this paper. Furthermore, perfect channel state information (CSI) is assumed to be available at the BS for resource allocation design to obtain a performance upper bound for OFDMA-URLLC MEC systems.
\end{remark}

\begin{figure}[t]
	\centering
	\scalebox{0.23}{
		\pstool{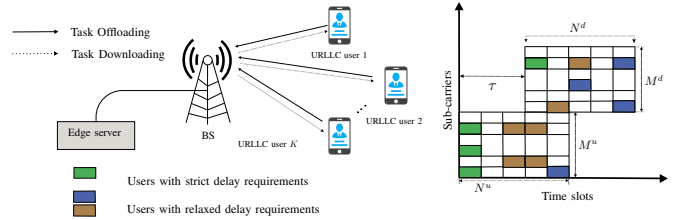}{
			\psfrag{e}[c][c][2]{Task Offloading}
			\psfrag{f}[c][c][2]{Sub-carriers}
			\psfrag{g}[c][c][2]{Task Downloading}
			\psfrag{a}[c][c][1.5]{URLLC user 1}	
			\psfrag{o}[c][c][2.5]{$\tau$}
			\psfrag{b}[c][c][1.5]{URLLC user 2}		
			\psfrag{c}[c][c][1.5]{\;\;URLLC user $K$}
			\psfrag{n}[c][c][1.8]{BS}
			\psfrag{s}[c][c][1.8]{Edge server}
			\psfrag{r}[c][c][2]{$\;\;$Users with relaxed delay requirements}
			\psfrag{h}[c][c][2]{Users with strict delay requirements}
			\psfrag{t}[c][c][2]{Time slots}
		    \psfrag{tu}[c][c][2]{$N^{u}$}
		    \psfrag{td}[c][c][2]{$N^{d}$}
		    \psfrag{wd}[c][c][2]{$M^{d}$}
	        \psfrag{wu}[c][c][2]{$M^{u}$}
	}}
	\caption{ Multiuser MEC system with a single BS with an edge server and $K$ URLLC users.}
	\label{model}
	\vspace{-0.5cm}
\end{figure}

\subsection{Uplink and Downlink Channel Models}
In the following, we introduce the uplink and downlink channel models for OFDMA-URLLC MEC systems. We assume that the channel gains of all users for all sub-carriers are constant during uplink and downlink transmission. 
In the uplink, the signal received at the BS from user $k$ on sub-carrier $m^{u}$ in time slot $n^{u}$ is given as follows:
\begin{equation}\label{su2a}\hspace{-0.65cm}
y^{u}_{k}[m^{u},n^{u}]=h_{k}^{u}[m^{u}]x^{u}_{k}[m^{u},n^{u}]+z^{u}_{BS}[m^{u},n^{u}],
\end{equation}
where $x^{u}_{k}[m^{u},n^{u}]$ denotes the symbol transmitted by user $k$ on sub-carrier $m^{u}$ in time slot $n^{u}$ to the BS. Moreover, $z^{u}_{BS}[m^{u},n^{u}]\sim \mathcal{CN}(0,\sigma^{2})$ denotes the noise at the BS\footnote{Without loss of generality, we assume that the noise processes at all receivers have identical variances.}, and $h^{u}_{k}[m^{u}]$ represents the complex channel coefficient between user $k$ and the BS on sub-carrier $m^{u}$. Moreover, for future use, we define the signal-to-noise ratio (SNR) of user $k$'s signal at the input of the BS's receiver on sub-carrier $m^{u}$ in time slot $n^{u}$ as follows:
\begin{equation}
\gamma^{u}_{k}[m^{u},n^{u}]=g^{u}_{k}[m^{u}]p^{u}_{k}[m^{u},n^{u}],
\end{equation}
where $p^{u}_{k}[m^{u},n^{u}]=\mathcal{E}\{|x^{u}_{k}[m^{u},n^{u}]|^{2}\}$ is the uplink transmitted power of user $k$ on sub-carrier $m^{u}$ in time slot $n^{u}$, and $g^{u}_{k}[m^{u}]=\frac{|h^{u}_{k}[m^{u}]|^{2}}{\sigma^{2}}$. A similar channel model is adopted for downlink transmission and the corresponding SNR at user $k$ on sub-carrier $m^{d}$ in time slot $n^{d}$ is denoted by $\gamma^{d}_{k}[m^{d},n^{d}]$.

\subsection{Achievable Rate for FBT}
Shannon's capacity theorem, on which most conventional resource allocation designs are based, applies to the asymptotic case where the packet length approaches infinity and the decoding error probability goes to zero \cite{shannon}. Thus, it cannot be used for resource allocation design for URLLC systems, as URLLC systems have to employ short packets to achieve low latency, which makes decoding errors unavoidable. For the performance evaluation of FBT, the so-called normal approximation for short packet transmission was developed in \cite{thesis}. For parallel complex AWGN channels, the maximum number of bits $\Psi$ conveyed in a packet comprising $L$ symbols can be approximated as follows\cite[Eq. (4.277)]{thesis},\cite[Fig. 1]{Erseghe1}:
\begin{align}\label{normalapproximation}
	\Psi=\sum_{l=1}^{L}\log_{2}(1+\gamma[l])-aQ^{-1}(\epsilon)\sqrt{\sum_{l=1}^{L}{V}[l]},
\end{align}
where $\epsilon$ is the decoding packet error probability, and $Q^{-1}(\cdot)$ is the inverse of the Gaussian Q-function with $
Q(x)=\frac{1}{\sqrt{2\pi}}\int_{x}^{\infty}\text{exp}{\left(-\frac{t^{2}}{2}\right)}\text{d}t$.
$V[l]=(1-{(1+\gamma[l])^{-2}})$
and $\gamma[l]$ are the channel dispersion \cite{thesis} and the SNR of the $l$-th symbol, respectively, and  $a=\log_{2}(\text{e})$. \color{black}

In this paper, we base the joint uplink-downlink resource allocation algorithm design for OFDMA-URLLC MEC systems on (\ref{normalapproximation}). By allocating several resource blocks from the available resources to a given user, the number of offloaded and downloaded bits of the user can be adjusted.
\section{Problem Formulation}
In this section, we explain the offloading and downloading process and introduce the QoS requirements of the URLLC MEC users. Moreover, we formulate the proposed resource allocation optimization problem.
\subsection{Offloading and Downloading}
The edge computing process is performed as follows. First, each user offloads its data to the edge server in the uplink. Subsequently, the edge server processes this data and sends the results back in the downlink to the user. Thus, uplink and downlink should satisfy the following constraints: 
\begin{align}&\label{offu1}\hspace{-0.25cm}
\mathrm{C1:}\Psi^{u}_{k}(\mathbf{s}_{k}^{u},\mathbf{p}_{k}^{u}) \geq B_{k},\forall k,\
\mathrm{C2:}\Psi^{d}_{k}(\mathbf{s}_{k}^{d},\mathbf{p}_{k}^{d})\geq \Gamma_{k}B_{k},\forall k,
\end{align} 
where
\begin{equation}\label{offu2}
\Psi^{j}_{k}(\mathbf{s}_{k}^{j},\mathbf{p}_{k}^{j})=	C^{j}_{k}(\mathbf{s}_{k}^{j},\mathbf{p}_{k}^{j})-V^{j}_{k}(\mathbf{s}_{k}^{j},\mathbf{p}_{k}^{j}), 
\end{equation}
and
\begin{align}\label{b2}\hspace{-0.5cm}
C^{j}_{k}(\mathbf{s}_{k}^{j},\mathbf{p}_{k}^{j})=\sum_{m^{j}=1}^{M^{j}}\sum_{n^{j}=1}^{N^{j }}s^{j}_{k}[m^{j},n^{j}]\log_{2}(1+\gamma^{j}_{k}[m^{j},n^{j}]),
\end{align} 
\begin{align}\label{b3}\hspace{-0.55cm}
V^{j}_{k}(\mathbf{s}^{j}_{k},\mathbf{p}_{k}^{j})=aQ^{-1}(\epsilon^{j}_{k})\sqrt{{\sum_{m^{j}=1}^{M^{j}}\sum_{n^{j}=1}^{N^{j}}s^{j}_{k}[m^{j},n^{j}]V^{j}_{k}[m^{j},n^{j}]}}.
\end{align} 
Here, $s^{j}_{k}[m^{j},n^{j}]=\{0,1\}, \forall m^{j},n^{j},$ are the sub-carrier assignment indicators. If sub-carrier $m^{j}$ in time slot $n^{j}$ is assigned
to user $k$, we have $s^{j}_{k}[m^{j},n^{j}]=1$, otherwise $s^{j}_{k}[m^{j},n^{j}]=0$. Furthermore, we assume that each sub-carrier is allocated to at most one user to avoid multiple access interference. $p^{j}_{k}[m^{j},n^{j}]$ is the power allocated to user $k$ on sub-carrier $m^{j}$ in time slot $n^{j}$. $\mathbf{s}_{k}^{j}$ and $\mathbf{p}_{k}^{j}$ are the collections of optimization variables $s^{j}_{k}[m^{j},n^{j}], \forall m^{j},n^{j}$, and $p^{j}_{k}[m^{j},n^{j}], \forall m^{j},n^{j}, \forall j$, respectively, and
$
V^{j}_{k}[m^{j},n^{j}]=(1-(1+\gamma^{j}_{k}[m^{j},n^{j}])^{-2})$. Constraints $\mathrm{C1}$ and $\mathrm{C2}$ guarantee for user $k$ the transmission of $B_{k}$ bits in the uplink and $\Gamma_{k}B_{k}$ bits in the downlink, respectively. Moreover, $\Gamma_{k},\forall k,$ is the ratio of the sizes of the computation results and the offloaded task. The value of $\Gamma_{k}$ depends on the application type, e.g., $\Gamma_{k}> 1$ is expected for augmented reality applications. \cite{jointsubchannel}. 
\subsection{Causality and Delay}
In the following, we explain the causality and delay constraints.

\textit{1) Causality:} According to Fig.~\ref{model}, downlink transmission cannot start for a given user before all data of this user has been received at the BS via the uplink. Thus, we impose the following causality constraints\footnote{In this paper, we neglect the computation time and power consumption at the edge server, and we only focus on uplink and downlink transmission. This model is valid when the edge server has sufficient processing and computation resources to carry out the small tasks of URLLC users.}:
\begin{align}&\label{Co}\hspace{-.5cm}
\mathrm{C3:}s_{k}^{u}[m^{u},\tau +o]+s^{d}_{k}[m^{d},n^{d}]\leq 1 ,\nonumber\\&\hspace{1
	cm} \forall o=\{1,\dots,\bar{O}\}, \forall k, \forall m^{u}, \forall n^{d}=\{1,\dots,o\}, \forall m^{d}.
\end{align}
This constraint ensures  that the downlink transmission for a particular user cannot start before its data has arrived at the BS. 

\textit{2) Delay:} The delay of a computation task is limited by requiring the downlink transmission to be finished before $D_{k}-\tau$ time slots as follows:
\begin{align}\hspace{-3cm}
\mathrm{C4:}s^{d}_{k}[m^{d},n^{d}] = 0, \forall n^{d} \geq D_{k}-\tau.
\end{align} 
The total latency of a computation task is determined by $D_{k}$ and $\tau$. Note that the values of $D_{k}$ and $\tau$ are known for resource allocation.
\subsection{Optimization Problem Formulation}
In the following, we formulate the resource allocation design problem with the objective to minimize the total weighted network power consumption, while satisfying the latency requirements for the users' task computation. In particular, we optimize the power and sub-carrier assignments in uplink and downlink. To this end, the optimization problem is formulated as follows:  
\allowdisplaybreaks
\begin{align}\label{Op1}
&\underset {\mathbf{s}^{u}, \mathbf{p}^{u}, \mathbf{s}^{d}, \mathbf{p}^{d}}{\text{min}}\sum_{k=1}^{K}w_{k}\sum_{m^{u}=1}^{M^{u}}\sum_{n^{u}=1}^{N^{u}}s^{u}_{k}[m^{u},n^{u}]p^{u}_{k}[m^{u},n^{u}]\\&\qquad\nonumber \qquad+\sum_{k=1}^{K}\sum_{m^{d}=1}^{M^{d}}\sum_{n^{d}=1}^{N^{d}}s^{d}_{k}[m^{d},n^{d}]p^{d}_{k}[m^{d},n^{d}]\\
&\text{s.t.} \; \mathrm{C1-C4}, \;\mathrm{C5:}\sum_{k=1}^{K}s^{u}_{k}[m^{u},n^{u}] \leq 1, \forall m^{u},n^{u},\nonumber\\&\quad\;\;
\mathrm{C6:} s^{u}_{k}[m^{u},n^{u}] \in \{0,1\}, \forall k,m^{u},n^{u},\nonumber\\&  \quad\;\;
\mathrm{C7:}  \sum_{m^{u}=1}^{M^{u}}\sum_{n^{u}=1}^{N^{u}}s^{u}_{k}[m^{u},n^{u}]p^{u}_{k}[m^{u},n^{u}] \leq P_{k,\text{max}}, \forall k, \nonumber\\&  \quad\;\;
\mathrm{C8:} p^{u}_{k}[m^{u},n^{u}] \geq 0, \forall k, m^{u}, n^{u},\nonumber\\& \quad\;\;
\mathrm{C9:}	\sum_{k=1}^{K}s^{d}_{k}[m^{d},n^{d}] \leq 1, \forall m^{d},n^{d},\nonumber\\&  \quad\;\; \mathrm{C10:} s^{d}_{k}[m^{d},n^{d}] \in \{0,1\}, \forall k,m^{d},n^{d}, \nonumber \nonumber\\&  \quad\;\; \mathrm{C11:} \sum_{k=1}^{K}\sum_{m^{d}=1}^{M^{d}}\sum_{n^{d}=1}^{N^{d}}s^{d}_{k}[m^{d},n^{d}]p^{d}_{k}[m^{d},n^{d}] \leq P_{\text{max}},\nonumber\\&  \quad\;\;
\mathrm{C12:} p^{d}_{k}[m^{d},n^{d}] \geq 0, \forall k, m^{d}, n^{d},\nonumber
\end{align}
where $\mathbf{s}^{j},\forall j,$ and $\mathbf{p}^{j},\forall j,$ are the collections of optimization variables $\mathbf{s}_{k}^{j}, \forall k,j$, and $\mathbf{p}_{k}^{j}, \forall k,j$, respectively. Moreover, $w_{k}\geq 1, \forall k,$ are weights that allow the prioritization of the uplink power consumption compared to the downlink power consumption. 

In (\ref{Op1}), constraints $\mathrm{C1}$ and $\mathrm{C2}$ guarantee the transmission of a minimum number of bits from user $k$ to the BS in the uplink and from the BS to user $k$ in the downlink, respectively. Constraint $\mathrm{C3}$ is the uplink-downlink causality constraint and constraint $\mathrm{C4}$ ensures that user $k$ is served within its delay requirements. Constraints $\mathrm{C5}$ and $\mathrm{C6}$ for the uplink and constraints $\mathrm{C9}$ and $\mathrm{C10}$ for the downlink are imposed to ensure that each sub-carrier in a given time slot is allocated to only one user.  Constraints $\mathrm{C7}$ and $\mathrm{C11}$ are the total power constraints for user $k$ and the BS, respectively. Constraints $\mathrm{C8}$ and $\mathrm{C12}$ are the  non-negative transmit power constraints. 

Optimization problem (\ref{Op1}) is a mixed-integer non-convex problem. The non-convexity has the following reasons. First, the optimization variables in the objective function and the constraints are coupled, e.g., $\mathrm{C1}$ and $\mathrm{C7}$. Second, the achievable rate for FBT has a non-convex structure. Finally, the integer constraints $\mathrm{C6, C10}$ are non-convex. In general, non-convex optimization problems cannot be solved optimally in polynomial time. Hence, in the next section, we focus on developing a sub-optimal solution, where the SCA method is employed for computational efficiency and real-time applicability.
\section{Solution of the Problem}
In this section, we first transform the problem in (\ref{Op1}) into a more tractable equivalent form. In particular, we first employ the Big-M formulation. Then, we use the difference of convex programming and SCA approaches in order to solve the optimization problem in (\ref{Op1}) iteratively. The main steps of the proposed low-complexity algorithm are summarized in Fig.~\ref{alg}.\hspace{-0.3cm}
 \begin{figure}[t!]	
	\centering
	\tikzstyle{decision} = [diamond, draw, fill=green!12, 
	text width=5em, text badly centered, node distance=3.5cm, inner sep=0pt]
	\tikzstyle{block} = [rectangle, draw, fill=blue!12, 
	text width=2.5cm, text centered, rounded corners, minimum height=1.5cm]
	\tikzstyle{line} = [draw, -latex']
	\tikzstyle{opt} = [rectangle,text badly centered,draw,fill=green!20, node distance=2.5cm, minimum height=1.7cm]
	\scalebox{.6}{
		\begin{tikzpicture} [node distance =3.2cm, auto] 
	\node[opt,node distance=8cm]  (a) {Non-convex Problem (\ref{Op1})};  
		\node[block,right of= a,node distance=4.2cm] (b) {1-Big-M Formulation}; 
		\node[block,right of=b] (c) {2-DC Reformulation};
		\node[block,right of=c] (d) {3-SCA \\ Sub-optimal Solution};
		\draw[line] (a)-- (b);  
		\draw[line] (b)-- (c);  
		\draw[line] (c)-- (d);  
		\end{tikzpicture}} 
	\caption{Illustration of the key steps of the proposed low-complexity scheme.}
\label{alg}\vspace{-0.3cm}
\end{figure}
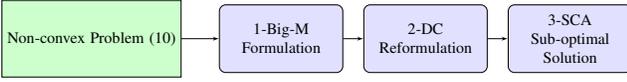\\
\subsection{Problem Transformation}
To deal with the non-convex product terms in optimization problem (\ref{Op1}), the Big-M method is employed\cite{yan}. 

\textbf{Step 1 (Big-M Formulation\footnote{ For more details on the big M-formulation,
		please refer to \cite[Section~2.3]{Leemixed}.}):} Let us introduce new optimization variables as 
\begin{align}&\label{prod}
\bar{p}_{k}^{j}[m^{j},n^{j}]=s_{k}^{d}[m^{j},n^{j}] {p}_{k}^{j}[m^{j},n^{j}], \forall k,m^{j},n^{j}, \forall j.
\end{align}
Now, we decompose the product terms above using the Big-M formulation (McCormick envelopes) and impose the following
additional constraints\cite{Leemixed}:
\begin{align}&
\mbox {C13}: \bar{p}^{u}_{k}[m^{u},n^{u}]\leq P_{k,\text{max}} s^{u}_{k}[m^{u},n^{u}], \forall k,m^{u},n^{u}, \\& 
\mbox {C14}: \bar{p}^{u}_{k}[m^{u},n^{u}]\leq p_{k}^{u}[m^{u},n^{u}], \forall k,m^{u},n^{u},\\ &  \mbox {C15}: \bar{p}_{k}^{u}[m^{u},n^{u}]\geq p_{k}^{u}[m^{u},n^{u}]\nonumber\\& \hspace{1.5cm} -(1-s_{k}^{u}[m^{u},n^{u}])P_{k,\text{max}}, \forall k,m^{u},n^{u},\quad  \\ & \mbox {C16}: \bar{p}_{k}^{u}[m^{u},n^{u}]\geq 0, \ \ \forall k,m^{u},n^{u},\quad \\ & \mbox {C17}: \bar{p}^{d}_{k}[m^{d},n^{d}]\leq P_{\text{max}} s^{d}_{k}[m^{d},n^{d}],  \forall k,m^{d},n^{d}, \\ &  \mbox {C18}: \bar{p}^{d}_{k}[m^{d},n^{d}]\leq p_{k}^{d}[m^{d},n^{d}], \ \ \ \forall k,m^{d},n^{d},\\ &  \mbox {C19}: \bar{p}_{k}^{d}[m^{d},n^{d}]\geq p_{k}^{d}[m^{d},n^{d}] \nonumber\\& \hspace{2.5cm}-(1-s_{k}^{d}[m^{d},n^{d}])P_{\text{max}}, \   \forall k,m^{d},n^{d},\quad \hspace{-4mm} \\ & \mbox {C20}: \bar{p}_{k}^{d}[m^{d},n^{d}]\geq 0, \ \ \forall k,m^{d},n^{d}.
\end{align}
The non-convex product terms $s_{k}^{d}[m^{j},n^{j}] {p}_{k}^{j}[m^{j},n^{j}], \forall k,m^{j},n^{j}, \forall j$ in (\ref{prod}) are transformed into a set of convex linear inequalities. Note that constraints $\mbox {C13-C20}$ do not change the feasible set.
Now, optimization problem (\ref{Op1}) is transformed into the following equivalent form: 
\begin{align}\label{Op2}
&\underset {\mathbf{s}^{u}, \mathbf{p}^{u}, \mathbf{s}^{d}, \mathbf{p}^{d},\bar{\mathbf{p}}^{u},\bar{\mathbf{p}}^{d}}{\text{min}}\sum_{k=1}^{K}w_{k}\sum_{m^{u}=1}^{M^{u}}\sum_{n^{u}=1}^{N^{u}}\bar{p}^{u}_{k}[m^{u},n^{u}]\\ &\qquad \qquad+\sum_{k=1}^{K}\sum_{m^{d}=1}^{M^{d}}\sum_{n^{d}=1}^{N^{d}}\bar{p}^{d}_{k}[m^{d},n^{d}]\nonumber\\
&\text{s.t.} \; \mathrm{C1:}	\bar{C}^{u}_{k}(\bar{\mathbf{p}}_{k}^{u})-\bar{V}^{u}_{k}(\bar{\mathbf{p}}_{k}^{u})\geq B_{k},\forall k,\nonumber\\ &\quad\;\;\mathrm{C2:}		\bar{C}^{d}_{k}(\bar{\mathbf{p}}_{k}^{d})-\bar{V}^{d}_{k}(\bar{\mathbf{p}}_{k}^{d})\geq \Gamma_{k}B_{k} ,\forall k,\nonumber\\ &\quad\;\;\mathrm{C3-C6,} \;
\mathrm{C7:}  \sum_{m^{u}=1}^{M^{u}}\sum_{n^{u}=1}^{N^{u}}\bar{p}^{u}_{k}[m^{u},n^{u}] \leq P_{k,\text{max}}, \forall k,\nonumber \\&  \quad\;\;
\mathrm{C8-C10,} \; \mathrm{C11:} \sum_{k=1}^{K}\sum_{m^{d}=1}^{M^{d}}\sum_{n^{d}=1}^{N^{d}}\bar{p}☻^{d}_{k}[m^{d},n^{d}] \leq P_{\text{max}},\nonumber \\&  \quad\;\;\mathrm{C12, C13-C20.}\nonumber
\end{align}
where \vspace{-0.35cm}
\begin{align}&\label{c1}
\bar{C}^{j}_{k}(\mathbf{\bar{p}}_{k}^{j})=\sum_{m^{j}=1}^{M^{j}}\sum_{n^{j}=1}^{N^{j }}\log_{2}(1+\bar{\gamma}^{j}_{k}[m^{j},n^{j}]),\\&
\bar{V}^{j}_{k}(\mathbf{\bar{p}}_{k}^{j})=aQ^{-1}(\epsilon^{j}_{k})\sqrt{\sum_{m^{j}=1}^{M^{j}}\sum_{n^{j}=1}^{N^{j}}\bar{V}^{j}_{k}[m^{j},n^{j}]},
\end{align} 
$\bar{\gamma}^{j}_{k}[m^{j},n^{j}]=g^{j}_{k}[m^{j}]\bar{p}^{j}_{k}[m^{j},n^{j}],$ and
$\bar{V}^{j}_{k}[m^{j},n^{j}]=(1-(1+\bar{\gamma}^{j}_{k}[m^{j},n^{j}])^{-2}).$ Moreover, $\mathbf{\bar{p}}_{k}^{j}$ is the collection of optimization variables $\bar{p}_{k}^{}[m^{j},n^{j}], \forall m^{j}, n^{j}, \forall j.$

Optimization problem (\ref{Op2}) is still non-convex. However, its structure is more tractable compared to problem (\ref{Op1}). In the following, we find a low-complexity solution to problem (\ref{Op2}) using the difference of convex programming and SCA methods.   
\subsection{Difference of Convex Programming}
\textbf{Step 2:} The two remaining difficulties for solving problem (\ref{Op2}) are the binary variables in constraints $\mathrm{C6}$ and $\mathrm{C10}$ and the structure of the achievable rate for FBT in $\mathrm{C1}$ and $\mathrm{C2}$. To tackle these issues, we employ a difference of convex (DC) programming approach\cite{ghanem1,yan,kwan1,Joinoptimization}. To this end, the integer constraints  in (\ref{Op2}) are rewritten in the following difference of convex function form:
\begin{align}\label{eq1}
&\mathrm{C6a:} 0 \leq s^{u}_{k}[m^{u},n^{u}] \leq 1, \forall k,m^{u},n^{u}, \\&
\mathrm{C6b:}  E^{u}(\mathbf{s}^{{u}}) -H^{u}(\mathbf{s}^{{u}}) \leq 0,
\end{align}
\begin{align}\label{eq2}
&\mathrm{C10a:} 0 \leq s^{d}_{k}[m^{d},n^{d}] \leq 1, \forall k,m^{d},n^{d}, \\&
\mathrm{C10b:}  E^{d}(\mathbf{s}^{{d}}) -H^{d}(\mathbf{s}^{{d}}) \leq 0,
\end{align}
where\vspace{-0.55cm}
\begin{align}&\label{esu}
E^{j}(\mathbf{s}^{{j}})=\sum_{k=1}^{K}\sum_{m^{j}=1}^{M^{j}}\sum_{n^{j}=1}^{N^{j}}s^{j}_{k}[m^{j},n^{j}], \forall j, \\& H^{j}(\mathbf{s}^{{j}})=\sum_{k=1}^{K}\sum_{m^{j}=1}^{M^{j}}\sum_{n^{j}=1}^{N^{j}}(s^{j}_{k}[m^{j},n^{j}] )^{2}, \forall j.    
\end{align} 
Now, constraints $\mathrm{C6}$, $\mathrm{C10}$ have been rewritten in continuous form, cf. $\mathrm{C6a}$, $\mathrm{C10a}$ . However, constraints $\mathrm{C6b}$, $\mathrm{C10b}$ are non-convex, i.e., reverse convex constraints. In order to handle them, we introduce the following lemma.
\begin{lem}
	For sufficiently large constant values $\eta_{1}$ and $\eta_{2}$ the optimization problem in (\ref{Op2})  is equivalent to the following problem:
	\begin{align}\label{Op3}& \underset {\mathbf{s}^{u}, \mathbf{p}^{u},\mathbf{s}^{d}, \mathbf{p}^{d},\mathbf{\bar{p}}^{u}, \mathbf{\bar{p}}^{d}}{ \mathop {\mathrm {minimize}}\nolimits }~\Phi(\mathbf{\bar{p}}^{u}, \mathbf{\bar{p}}^{d})+\eta_{1}(E^{u}-H^{u})+\eta_{2}(E^{d}-H^{d})
		\\&\;\; \mathrm {s.t.}~ \mathrm{C1-C5, C6a, C7-C9, C10a, C11-C20},\nonumber
	\end{align} 
where $\Phi(\mathbf{\bar{p}}^{u}, \mathbf{\bar{p}}^{d})$ is the objective function of problem (\ref{Op2}).    
\end{lem}           
\begin{proof}
Please refer to Appendix A.
\end{proof}
The only remaining sources of non-convexity are the structure of the achievable rate for FBT and the non-convex objective function. In the following, we employ SCA to approximate problem (\ref{Op3}) by a convex problem. Subsequently, we propose an iterative algorithm to find a low-complexity solution to problem (\ref{Op3}).      
\subsection{Successive Convex Approximation}
\textbf{Step 3:} In order to cope with the remaining non-convexity of  (\ref{Op3}), we employ the Taylor series approximation to approximate the non-convex parts of the objective function and constraints $\mathrm{C1}$ and $\mathrm{C2}$. Since $H^{j}(\mathbf{s}^{{j}}),\forall j,$ and $-\bar{V}^{j}_{k}(\mathbf{\bar{p}}^{j}_{k}),\forall j,$ are differentiable convex functions, then for any feasible points $\mathbf{s}^{{j(i)}}, \mathbf{\bar{p}}^{j(i)}_{k}, \forall j$, the following inequalities hold:          
\begin{align}&\label{htayslorhu}\hspace{-1.2cm}
H^{j}(\mathbf{s}^{{j}}) \ge\bar{H}^{j}(\mathbf{s}^{{j}})= H^{j}(\mathbf{s}^{{j(i)}}) \nonumber \\& +\nabla _{\mathbf {s}^{j}} H^{j}(\mathbf {s}^{j(i)})^{T}(\mathbf{s}^{{j}}-\mathbf{s}^{{j(i)}}), \forall j, 
\end{align}        
and \vspace{-0.25cm}
\begin{align}\label{vua}
\bar{V}^{j}_{k}(\mathbf{\bar{p}}^{j}_{k}) \leq \tilde{V}^{j}_{k}(\mathbf{\bar{p}}^{j}_{k},\mathbf{\bar{p}}^{j(i)}_{k}) =  \bar{V}^{j}_{k}(\mathbf{\bar{p}}^{j(i)}_{k})\nonumber\\&\hspace{-3.2cm}+ \nabla_{\mathbf{\bar{p}}^{j}_{k}}{\bar{V}}_{k}(\mathbf{\bar{p}}^{j(i)}_{k})^{T}(\mathbf{\bar{p}}^{j}_{k}-\mathbf{\bar{p}}^{j(i)}_{k}), \forall j.
\end{align}
The right hand sides of (\ref{htayslorhu}) and (\ref{vua}) are affine functions representing the global underestimation of $H^{j}(\mathbf{s}^{{j}}), \forall j$, and $\bar{V}^{j}_{k}(\mathbf{\bar{p}}^{j}_{k}),\forall j$, respectively, where $\nabla _{\mathbf {s}^{j}} H^{j}(\mathbf {s}^{j(i)})^{T}(\mathbf{s}^{{j}}-\mathbf{s}^{{j(i)}})$ and $\nabla_{\mathbf{\bar{p}}^{j}_{k}}{\bar{V}}^{j}_{k}(\mathbf{\bar{p}}_{k}^{j({i})})$ are given on the top of the next page.  
\begin{figure*}[t!]\vspace{-1.5cm}
	 \centering
\begin{align}&
\nabla_{\mathbf{s}^{j}}H^{j}(\mathbf{s}^{j({i})})^{T}(\mathbf{s}^{j}-\mathbf{s}^{j({i})})  =\sum_{k=1}^{K}\sum_{m^{j}=1}^{M^{j}}\sum_{n^{j}=1}^{N^{j}}2s^{j({i})}_{k}[m^{j},n^{j}]\left( s^{j}_{k}[m^{j},n^{j}]-s^{j(i)}_{k}[m^{j},n^{j}]\right),\forall j,\\&\vspace{-0.5cm} \nabla_{\mathbf{\bar{p}}_{k}^{j}}\bar{V}^{j}_{k}(\mathbf{\bar{p}}_{k}^{j({i})})
= \frac{aQ^{-1}(\epsilon^{j}_{k})}{\sqrt{\sum_{m^{j}=1}^{M^{j}}\sum_{n^{j}=1}^{N^{j}} \bar{V}^{j(i)}_{k}[m^{j},n^{j}]}}\begin{pmatrix} 
\ \frac{g^{j}_{k}[1]}{(1+\bar{p}^{j(i)}_{k}[1,1]g^{j}_{k}[1])^{3}} \\
\vdots \\
\frac{g_{k}^{j}[M]}{(1+\bar{p}_{k}^{j(i)}[M,N]g_{k}^{j}[M])^{3}} \end{pmatrix}, \forall j.
\end{align}
\noindent \rule{18.5cm}{0.5pt}
\end{figure*}    
 By substituting the right hand sides of (\ref{htayslorhu}) and (\ref{vua}) into (\ref{Op3}), we obtain the following optimization problem:    
\begin{align}&\label{op3a}
\underset {\mathbf{s}^{u}, \mathbf{p}^{u},\mathbf{s}^{d}, \mathbf{p}^{d},\mathbf{\bar{p}}^{u}, \mathbf{\bar{p}}^{d}}{ \mathop {\mathrm {minimize}}\nolimits }~\Phi(\bold{\bar{p}}^{u}, \bold{\bar{p}}^{d})+\eta_{1}(E^{u}-\bar{H}^{u})+\eta_{2}(E^{d}-\bar{H}^{d})	
\\&\;\; \mathrm {s.t.}~\mathrm{C1:}	C^{u}_{k}(\mathbf{\bar{p}}_{k}^{u})-\tilde{V}^{u}_{k}(\mathbf{\bar{p}}_{k}^{u}, \mathbf{\bar{p}}^{u(i)}_{k})\geq B_{k},\forall k, \nonumber\\& \;\qquad \mathrm{C2:}		C^{d}_{k}(\mathbf{\bar{p}}_{k}^{d})-\tilde{V}^{d}_{k}(\mathbf{\bar{p}}_{k}^{d}, \mathbf{\bar{p}}^{d(i)}_{k})\geq \Gamma_{k}B_{k} \nonumber,\forall k, \\& \;\qquad \mathrm{C3-C5, C6a, C7-C9, C10a, C11-C20}.\nonumber
\end{align}  
Optimization problem (\ref{op3a}) is convex because the objective function is convex and the constraints span a convex set. Therefore, it can be efficiently solved by standard convex optimization solvers such as CVX \cite{cvx}. Algorithm 1 summarizes the main steps to solve (\ref{Op3}) in an iterative manner, where the
solution of (\ref{op3a}) in iteration ($i$) is used as the initial point for
the next iteration $(i+1)$. The algorithm produces a sequence of
improved feasible solutions until convergence to a local optimum
point of problem (\ref{Op3}) or equivalently problem (\ref{Op1}) in polynomial time.
 \begin{algorithm}[t]
	\caption{Successive Convex Approximation}
	1: {Initialize:} Random initial points $\mathbf{s}^{u(1)}$, $\mathbf{s}^{d(1)}$, $\mathbf{\bar{p}}^{u(1)}$, $\mathbf{\bar{p}}^{d(1)}$, set iteration index $i=1$, maximum number of iterations $I_{\text{max}}$, and initial penalty factors, $\eta_{1} >0$ and $\eta_{2} >0$.\\
	2: \textbf{Repeat}\\
	3: Solve convex problem (\ref{op3a}) for given  $\mathbf{s}^{u(i)}$, $\mathbf{s}^{d(i)}$, $\mathbf{\bar{p}}^{u(i)}$, $\mathbf{\bar{p}}^{d(i)}$, and store the intermediate solutions   $\mathbf{s}^{u}$, $\mathbf{s}^{d}$, $\mathbf{\bar{p}}^{u}$, $\mathbf{\bar{p}}^{d}$\\
	4: Set ${i}={i}+1$ and update $\mathbf{s}^{u(i)}=\mathbf{s}^{u}$, $\mathbf{s}^{d(i)}=\mathbf{s}^{d}$,  $\mathbf{\bar{p}}^{u(i)}=\mathbf{\bar{p}}^{u}$,
	$\mathbf{\bar{p}}^{d(i)}=\mathbf{\bar{p}}^{d}$. \\
	6: \textbf{Until} convergence or $i=I_{\text{max}}$.\\
	7: {Output:} $\mathbf{s}^{u*}=\mathbf{s}^{u}$,
	$\mathbf{s}^{d*}=\mathbf{s}^{d}$,
	$\mathbf{\bar{p}}^{u*}=\mathbf{\bar{p}}^{u}$,
	$\mathbf{\bar{p}}^{d*}=\mathbf{\bar{p}}^{d}$.
	\label{sca}
\end{algorithm}  
\section{Performance Evaluation}
In this section, we provide simulation results to evaluate the effectiveness of the proposed joint uplink-downlink resource allocation algorithm for OFDMA-URLLC MEC systems. We adopt the simulation parameters given in Table I, unless specified otherwise. In our simulations, a single cell is considered with inner and outer radii $r_{1}=50~\textrm{m}$ and $r_{2}=100~\textrm{m}$, respectively. The BS is located at the center of the cell, and the users are randomly located between the inner and the outer radii. The user weights are set to  $w_{k}=1,\forall k$ for simplicity. The path loss is calculated as  $35.3 + 37.6 \log_{10}(d_{k})$\cite{chsecross}, where $d_{k}$ is the distance from the BS to user $k$. The values of the penalty factors are set to $\eta_{1}=10KP_{k, \text{max}}$  and $\eta_{2}=10P_{\text{max}}$. The small scale fading gains between the BS and the users are modeled as independent and identically Rayleigh distributed. All simulation results are averaged over $100$ realizations of the path loss and multipath fading.
\begin{table}[t]
	\label{tab:table}
	\centering
	\caption{Simulation Parameters.} 
	\renewcommand{\arraystretch}{1.4}
	\scalebox{0.5}{%
		\begin{tabular}{|c||c|}
			\hline
			Parameter & Value \\ \hline \hline 
			Total number  of sub-carriers in uplink and downlink $M=M^{u}=M^{d}$ & $2M$=64 \\ \hline
			Number  of time slots in uplink and downlink $N^{u}=N^{d}$ & 4 \\ \hline
		    Bandwidth of each sub-carrier & 30 kHz \\ \hline
			Noise power density  & -174 dBm/Hz \\ \hline
			Maximum BS transmit power, $P_{\text{max}}$  &  $45$~dBm \\ \hline  
		Maximum transmitted power of each user, $P_{k, \text{max}}$  &  $23$~dBm \\ \hline  
			Value of $\Gamma_{k},\forall k$ &  1						\\ \hline 			
	\end{tabular}}
\end{table} 
\subsection{Performance Bound and Benchmark Scheme}
We compare the performance of the proposed resource allocation algorithm with the following benchmark schemes:
\begin{itemize}
	\setlength{\itemsep}{1pt}
\item {\textbf{Shannon's capacity (SC)}}: To obtain an (unachievable) lower bound on the total network power consumption, Shannon's capacity formula is adopted in problem (\ref{Op1}), i.e., $V^{j}_{k}(\mathbf{s}_{k}^{j},\mathbf{p}_{k}^{j}), \forall j,$ is set to zero in constraints $\mbox {C1}$ and $\mbox {C2}$, respectively, and all other constraints are retained. The resulting optimization problem is solved using a modified version of the proposed algorithm. 
		\setlength{\itemsep}{1pt}
	\item {\textbf{Fixed sub-carrier assignment (FSA)}}: In this scheme, we fix the sub-carrier assignment. In fact, we divide the total number of sub-carriers among the users such that their delay and causality constraints are met. Then, we  optimize the power allocated to the sub-carriers for the given channel realization. The resulting optimization problem is solved using the SCA method.
\end{itemize}
\subsection{Simulation Results}
In Fig.~\ref{bits}, we investigate the average system power consumption versus the size of the task of the URLLC users and study the impact of different delay requirements. For delay scenario $S_{0}$, none of the users has delay restrictions, i.e., $D_{k}=\tau+N^{d}= 7, \forall k$. In contrast, for delay scenario $S_{1}$, two users have strict delay constraints while the remaining users do not, i.e., $D_{1} = D_{2}= 5$ and $D_{3} = D_{4} = 7$. As expected, increasing the required number of transmitted bits leads to higher transmit powers. This is due to the fact that if more bits are to be transmitted in a given frame, higher SNRs are needed, and thus, the BS and the users have to increase the transmitted power. Furthermore, the proposed scheme leads to a substantially lower power consumption compared to the FSA scheme. This is due to the non-optimal sub-carrier allocation for the FSA scheme. Fig.~\ref{bits} also reveals the impact of strict delay requirements. In particular, delay scenario $S_{1}$ leads to a higher power consumption compared to $S_{0}$ because the BS and the users are forced to allocate more power even if their channel conditions are poor to ensure their transmissions are completed with the desired delay. Furthermore, SC provides a lower bound for the required power consumption of OFDMA-URLLC MEC systems. However, SC cannot guarantee the required latency and reliability. This is due to the fact that, in this scheme, the performance loss incurred by FBT is not taken into account for resource allocation design, and thus the obtained resource allocation policies may not
meet the QoS constraints.

In Fig.~\ref{Pe}, we show the average system power consumption versus the packet error probability and study the impact of different delay requirements. As can be observed, for the proposed scheme and FSA, the average system power consumption is a monotonically decreasing function of the packet error probability. This is due to the fact that the complementary error function in the normal approximation is a monotonically decreasing function of $\epsilon$, and as a result, the impact of the dispersion part in the normal
approximation decreases as $\epsilon$ increases. Fig.~\ref{Pe} also reveals the impact of delay constraints. In particular, delay scenario $\bar{S}_{1}=\{D_{1}=D_{2}=D_{3}=5$, $D_{4}=D_{5}=7\}$ leads to a higher power consumption compared to $\bar{S}_{0}=\{D_{k}=7, \forall k\}$. This is due to the smaller feasible set of the optimization problem. Moreover, as can be seen, for SC, the power consumption is independent of the packet error probability. This is due to the fact that SC assumes that the decoding error probability is zero. Moreover, the gap between the proposed scheme and SC is the price to be paid for enforcing strict delay and reliability requirements to ensure URLLC.
\begin{figure}[t!]
	\centering
	\resizebox{0.99\linewidth}{!}{\psfragfig{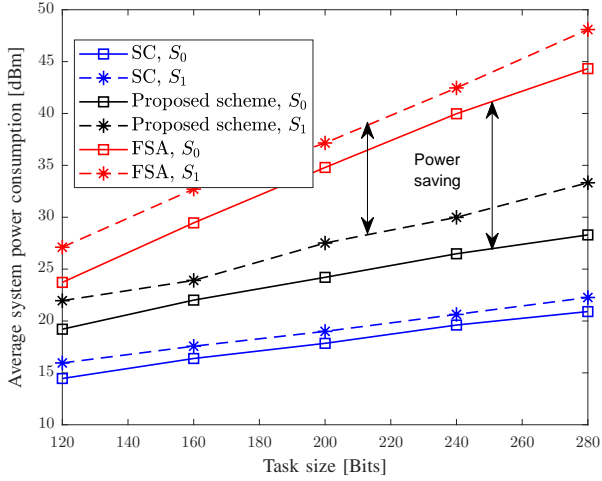}}\vspace{-4mm}
	\caption{Average consumed power [dBm] vs. task size [bits], $K=4$, $ \tau=3, \bar{O}=1$, $\epsilon^{j}_{k}=10^{-6}, \forall j, k.$} 
	\label{bits}\vspace{-0.5cm}
\end{figure}
\begin{figure}[t!]
	\centering
	\resizebox{0.99\linewidth}{!}{\psfragfig{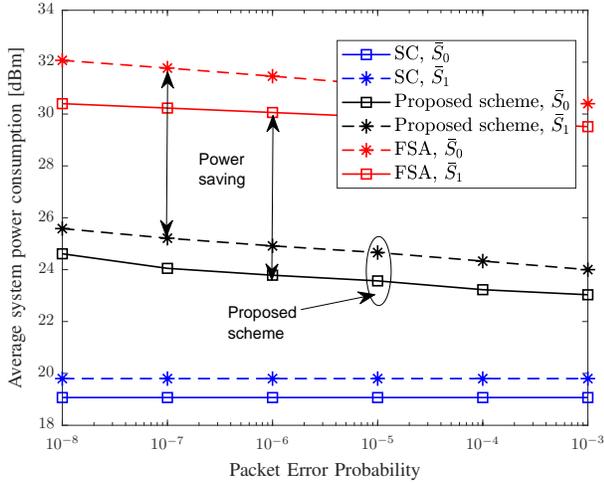}}\vspace{-4mm}
\caption{Average consumed power [dBm] vs. packet error probability, $K=5$, $ \tau=3, \bar{O}=1$, $B_{k}=160~\text{bits}, \forall k$.} 
	\label{Pe}
\end{figure}

\section{Conclusions}
This paper studied the resource allocation algorithm design for OFDMA-URLLC MEC systems. To ensure the stringent end-to-end transmission delay and reliability requirements of URLLC, we proposed a joint uplink-downlink resource allocation scheme which takes into account FBT. Moreover, to minimize the end-to-end delay, we proposed a partial time overlap between the uplink and downlink frames which introduces new uplink-downlink causality constraints. The proposed resource allocation algorithm design was formulated as an optimization problem for minimization of the total weighted transmit power of the network under QoS constraints regarding the minimum required number of computed bits of the URLLC users within a maximum computing time, i.e., the end-to-end delay. Due to the non-convexity of the formulated problem, finding a global solution entails a prohibitive computational complexity. Thus, a low-complexity algorithm based on SCA was proposed to find a high-quality sub-optimal solution. Our simulation results showed that the proposed resource allocation algorithm design facilitates the application of URLLC in MEC systems, and achieves significant power savings compared to a benchmark scheme.
	\section*{Appendix A}
The proof follows similar steps as corresponding proofs in \cite{ghanem1, yan,kwan1}. In the following, we show that problems (\ref{Op3}) and (\ref{Op2}) are equivalent. Let $U^{*}$ denote the optimal objective value of (\ref{Op3}). We define the Lagrangian function, denoted by ${\mathcal {L}}(\bar{\mathbf {p}}^{u},\bar{\mathbf {p}}^{d},\mathbf {s}^{u},\mathbf {s}^{d},\eta_{1},\eta_{2})$, as \cite{Boyed}
\begin{align}&\label{La} {\mathcal {L}}(\bar{\mathbf {p}}^{u},\bar{\mathbf {p}}^{d},\mathbf {s}^{u},\mathbf {s}^{d},\eta_{1},\eta_{2})=\nonumber\\& \hspace{1cm}\Phi(\mathbf{\bar{p}}^{u}, \mathbf{\bar{p}}^{d})+\eta_{1}(E^{u}-H^{u})+\eta_{2}(E^{d}-H^{d}),
\end{align}
where $\eta_{1}$ and $\eta_{2}$ are the Lagrange multipliers corresponding to constraints $\mbox {C6b}$ and $\mbox {C10b}$, respectively.  
Note that $E^{u}(\mathbf {s}^{u})-H^{u}(\mathbf {s}^{u}) \geq 0$ and $E^{d}(\mathbf {s}^{d})-H^{d}(\mathbf {s}^{d}) \geq 0$ hold.
Using Lagrange duality \cite{Boyed,Joinoptimization,ghanem1}, we have the following relation \footnote{Note that weak duality holds for convex and non-convex optimization problems\cite{Boyed}.}
\begin{align}\label{eq28} &U_{d}^{*}=\underset {\eta_{1}, \eta_{2} \ge 0}{ \mathop {\mathrm {max}}\nolimits } \quad \underset {\mathbf {p}^{u}, \mathbf {p}^{d},\mathbf {s}^{u},\mathbf {s}^{d}, \bar {\mathbf {p}}^{u},\bar{\mathbf {p}}^{d} \in \boldsymbol{\Omega }}{ \mathop {\mathrm {min}}\nolimits } \quad {\mathcal {L}}(\bar{\mathbf {p}}^{u},\bar{\mathbf {p}}^{d},\mathbf {s}^{u},\mathbf {s}^{d},\eta_{1},\eta_{2}) \IEEEyesnumber \IEEEyessubnumber
	\\\overset {(a)}{\le }&\label{eq29}\underset {\mathbf {p}^{u}, \mathbf {p}^{d},\mathbf {s}^{u},\mathbf {s}^{d}, \bar {\mathbf {p}}^{u},\bar{\mathbf {p}}^{d} \in \boldsymbol{\Omega }}{ \mathop {\mathrm {min}}\nolimits } \quad \underset {\eta_{1},\eta_{2} \ge 0}{ \mathop {\mathrm {max}}\nolimits } \quad {\mathcal {L}}(\bar{\mathbf {p}}^{u},\bar{\mathbf {p}}^{d},\mathbf {s}^{u},\mathbf {s}^{d},\eta_{1},\eta_{2}) = U^{*},\IEEEyessubnumber
\end{align}
where $\boldsymbol{\Omega }$ is the feasible set specified by the constraints in (\ref{Op3}).  
In the following, we first prove the strong duality, i.e., $U_{d}^{*}=U^{*}$. 
Let $(\mathbf {p}^{u*}, \mathbf {p}^{d*},\mathbf {s}^{u*},\mathbf {s}^{d*}, \bar {\mathbf {p}}^{u*},\bar{\mathbf {p}}^{d*},\eta_{1}^{*},\eta_{2}^{*})$ denotes the solution of (\ref{eq28}). For this solution, the following two cases are possible. \textit{Case 1)} If $E^{u}(\mathbf {s}^{u})-H^{u}(\mathbf {s}^{u})>0$ and $E^{d}(\mathbf {s}^{d})-H^{d}(\mathbf {s}^{d})>0$ hold, the optimal $\eta_{1}^{*}$ and $\eta_{2}^{*}$ are infinite, respectively. Hence, $U_{d}^{*}$ is infinite too, which contradicts the fact that it is upper bounded by a finite-value $U^{*}$.  \textit{Case 2)} If $E^{u}(\mathbf {s}^{u})-H^{u}(\mathbf {s}^{u})=0$ and $E^{d}(\mathbf {s}^{d})-H^{d}(\mathbf {s}^{d})=0$ hold, then $(\mathbf {p}^{u*}, \mathbf {p}^{d*},\mathbf {s}^{u*},\mathbf {s}^{d*}, \bar {\mathbf {p}}^{u*},\bar{\mathbf {p}}^{d*})$ belongs to the feasible set of the original problem (\ref{Op2}) which implies $U_{d}^{*}=U^{*}$. Hence, strong duality holds, and we can focus on solving the dual problem (\ref{eq28}) instead of the primal problem (\ref{eq29}). 

Next, we show that any $\eta_{1} \geq \eta_{1,0}$ and $\eta_{1} \geq \eta_{1,0}$ are optimal solutions for dual problem (\ref{eq28}), i.e., $\eta_{1}^{*}$ and $\eta_{2}^{*}$, where $\eta_{1,0}$ and $\eta_{2,0}$ are some sufficiently large numbers. To do so, we show that ${\Theta}(\eta_{1},\eta_{2})\triangleq \underset {\bar{\mathbf {p}}^{u},\bar{\mathbf {p}}^{d},{\mathbf {p}}^{u},{\mathbf {p}}^{u},{\mathbf {s}}^{u},{\mathbf {s}}^{d}\in \boldsymbol{\Omega }}{ \mathop {\mathrm {min}}\nolimits} \quad {\mathcal {L}}(\bar{\mathbf {p}}^{u},\bar{\mathbf {p}}^{d},\mathbf {s}^{u},\mathbf {s}^{d},\eta_{1},\eta_{2})$ is a monotonically increasing function of $\eta_{1}$ and $\eta_{2}$. Recall that $E^{u}(\mathbf {s}^{u})-H^{u}(\mathbf {s}^{u}) \geq 0$ and $E^{d}(\mathbf {s}^{d})-H^{d}(\mathbf {s}^{d}) \geq 0$ holds for any given ${\mathbf {p}}^{u},{\mathbf {p}}^{d},{\mathbf {s}}^{u},{\mathbf {s}}^{d},\bar{\mathbf {p}}^{u},\bar{\mathbf {p}}^{d}\in \boldsymbol{\Omega }$. Therefore, ${\mathcal {L}}(\bar{\mathbf {p}}^{u},\bar{\mathbf {p}}^{d},\mathbf {s}^{u},\mathbf {s}^{d},\eta_{1}(1),\eta_{2}(1)) \leq {\mathcal {L}}(\bar{\mathbf {p}}^{u},\bar{\mathbf {p}}^{d},\mathbf {s}^{u},\mathbf {s}^{d},\eta_{1}(2),\eta_{2}(2))$ holds for any given $\bar{\mathbf {p}}^{u},\bar{\mathbf {p}}^{d},{\mathbf {p}}^{u},{\mathbf {p}}^{u},{\mathbf {s}}^{u},{\mathbf {s}}^{d}\in \boldsymbol{\Omega }$, $0 \leq \eta_{1}(1)\leq \eta_{1}(2)$, and $0 \leq \eta_{2}(1)\leq \eta_{2}(2)$. This implies that ${\Theta}(\eta_{1}(1),\eta_{2}(1))\leq {\Theta}(\eta_{1}(2),\eta_{2}(2))$   and that ${\Theta}(\eta_{1},\eta_{2})$ is monotonically increasing in $\eta_{1}$ and $\eta_{2}$. Using this result, we can conclude that  
$\Theta (\eta_{1},\eta_{2}) =U^{*},\, \forall \eta_{1} \ge  \eta_{1,0}$, $\eta_{2} \ge  {\eta}_{2,0}$.

In summary, due to strong duality, we can use the dual problem  (\ref{Op3}) to find the solution of the primal problem (\ref{Op2}) and any $\eta_{1} \geq \eta_{1,0}$ and $\eta_{1} \geq \eta_{1,0}$ are optimal dual variables. These results are concisely given in Lemma 1 which concludes the proof.
\bibliography{ref}  
\bibliographystyle{IEEEtran}
\end{document}